\documentclass[journal,hidelinks]{IEEEtran}

\usepackage{amsmath}
\usepackage{amsfonts}
\usepackage{amssymb}
\usepackage{bm}

\usepackage{cite}
\usepackage{acronym}

\usepackage{graphicx}
\usepackage[caption=false,font=footnotesize]{subfig}

\usepackage{algorithm}
\usepackage{algpseudocode}

\algrenewcommand\algorithmicrequire{\textbf{Input:}}
\algrenewcommand\algorithmicensure{\textbf{Output:}}

\usepackage{xcolor}

\definecolor{myblue}{RGB}{30,100,180}
\definecolor{myred}{RGB}{190,45,45}
\definecolor{linkblue}{RGB}{32,105,185}
\definecolor{linkred}{RGB}{195,48,48}
\definecolor{panelblue}{RGB}{55,110,175}
\definecolor{trackgray}{RGB}{105,110,120}

\usepackage{tikz}
\usetikzlibrary{
    arrows.meta,
    calc,
    positioning,
    decorations.pathreplacing,
    patterns,
    backgrounds
}

\usepackage[hidelinks]{hyperref}
\usepackage{orcidlink}

\renewcommand{\Re}{\mathfrak{Re}}

\newtheorem{theorem}{Theorem}
\newtheorem{lemma}{Lemma}
\newtheorem{proposition}{Proposition}
\newtheorem{corollary}{Corollary}
\newtheorem{remark}{Remark}

\makeatletter
\def\@begintheorem#1#2{%
    \trivlist
    \item[\hskip\labelsep{\bfseries #1\ #2:}]%
    \itshape
}
\def\@opargbegintheorem#1#2#3{%
    \trivlist
    \item[\hskip\labelsep{\bfseries #1\ #2\ (#3):}]%
    \itshape
}
\makeatother

\newcommand{\vect}[1]{\mathbf{#1}}
\newcommand{\T}{\mathsf{T}}
\newcommand{\Hm}{\mathsf{H}}
\newcommand{\gc}{\gamma_{\rm c}}
\newcommand{\Gt}{G_{\rm t}}
\newcommand{\gcb}{\bar\gamma_{\rm c}}
\newcommand{\Gtb}{\bar G_{\rm t}}
\newcommand{\hu}{\vect{h}_{\rm u}}
\newcommand{\htt}{\vect{h}_{\rm t}}
\newcommand{\Ds}{\Delta_s}

\begin{document}
\title{Movable Beyond-Diagonal RIS for ISAC:\\ Optimal Placement and Pareto Boundary}

\author{Chandan~Kumar~Sheemar,~\IEEEmembership{Member,~IEEE,} Sourabh Solanki~\IEEEmembership{Member,~IEEE,}  and Symeon~Chatzinotas~\IEEEmembership{Fellow,~IEEE,}
\thanks{C.K. Sheemar and S. Chatzinotas are with the Interdisciplinary Centre for Security, Reliability and Trust (SnT), University of Luxembourg, L-1855 Luxembourg (e-mail: chandankumar.sheemar@uni.lu; symeon.chatzinotas@uni.lu). S. Solanki is with the National Institute of Technology Warangal, TS 506004, India (e-mail: ssolanki@nitw.ac.in); }}

\maketitle
 
 \begin{abstract}
This letter investigates integrated sensing and communication
using a movable beyond-diagonal reconfigurable intelligent
surface (BD-RIS). Under a rank-one line-of-sight BS--RIS
channel, we derive the exact communication--sensing Pareto
boundary of a fully-connected BD-RIS in terms of
user--target channel correlation. For line-of-sight RIS
links, maximizing this correlation yields the optimal
placement. We characterize the minimum track aperture
that eliminates the illumination trade-off and show that
a movable diagonal RIS then attains the same optimum.
Below this threshold, a dynamic program offers certified
placement accuracy. Numerical results examine hardware
connectivity, angular mismatch, and Rician fading.
\end{abstract}


\section{Introduction}
\IEEEPARstart{I}{ntegrated} sensing and communication (ISAC) lets one waveform and one aperture serve users and probe targets~\cite{liu2022isac}. Beyond-diagonal reconfigurable intelligent surfaces (BD-RISs) interconnect their elements through a reconfigurable impedance network, which generalizes the diagonal phase-shift model to symmetric unitary scattering matrices~\cite{shen2022modeling,nerini2024closed}. BD-RIS-aided ISAC outperforms diagonal RIS in sensing SNR and throughput~\cite{esmaeilbeig2025bd,liu2024enhancing}, but with elements at fixed positions and through alternating optimization.

A fixed surface, however, faces a limitation that no choice of scattering matrix can remove. When the BS reaches the surface over a line-of-sight link, this link has rank one, and the BS can only scale the incident field, not shape it. The surface must then serve the user and the target from a single field. With $\lambda/2$ spacing, the steering vectors toward two well-separated directions are nearly orthogonal, and their correlation decays as the surface grows. In this case each function can be favored only at the expense of the other. Interconnections widen the set of reflected fields but do not change this geometry, and adding elements makes it worse. Moving the elements changes the geometry itself, which raises two questions this letter answers: \emph{where should the elements go, and how much hardware do we still need once they are placed?} The second question also has a cost dimension: a fully connected (FC) BD-RIS needs $M(M+1)/2$ tunable impedances, whereas a diagonal RIS needs only $M$.

The element position is a second degree of freedom~\cite{zhu2024movable}. RISs with movable elements improve ISAC~\cite{ye2026joint,peng2026exploiting}. Xu \emph{et al.}~\cite{xu2026movable} combined movement with BD interconnections and found numerically, for the multiuser sum rate, that moving and interconnecting each win in a different regime. \emph{For ISAC, it is not yet known how the BD-RIS elements should be placed, nor whether moving and interconnecting complement or replace each other.}  

This letter answers these questions for a movable BD-RIS that serves one user and illuminates one target. Our contributions are:
\begin{itemize}
\item We show that an FC BD-RIS turns the surface into a field synthesizer, so the Pareto boundary takes the two-direction form known from MISO beamforming~\cite{jorswieck2008complete} and is governed only by the user--target correlation $\rho$. For group-connected and diagonal RISs, the supported part of the boundary follows from an exact two-dimensional search.
\item Under LoS, placements with larger $\rho$ Pareto-dominate those with smaller $\rho$, so maximizing $\rho$ is optimal for the FC BD-RIS. The boundary reduces to the single-function optimum iff the aperture exceeds $D^\star=(M-1)\lambda/|\Ds|$, using a lattice that places a grating lobe on the target; there, a movable diagonal RIS is also optimal. Below $D^\star$, a dynamic program (DP) with a certified accuracy bound and a closed-form clustered lattice give the placement.
\item Simulations validate the analysis, compare against a fixed sparse array of the same aperture under geometry mismatch, and evaluate Rician fading with an absolute user-SNR requirement.
\end{itemize}

\emph{Notation:} $(\cdot)^*$, $(\cdot)^\T$, $(\cdot)^\Hm$ denote conjugate, transpose and Hermitian transpose; $\vect{v}_{\mathcal G}$ is the subvector of $\vect{v}$ indexed by $\mathcal G$; $[a]^+=\max(a,0)$.

\section{System Model}

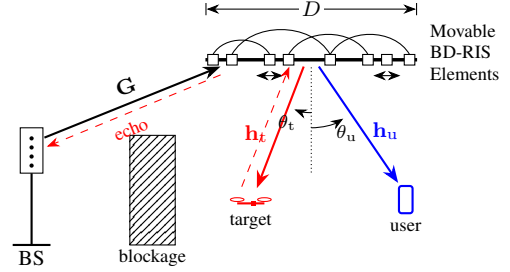
\begin{figure}[t]
\centering
\begin{tikzpicture}[x=1cm,y=1cm,>=Stealth,font=\footnotesize]
\fill[pattern=north east lines] (1.75,0) rectangle (2.35,1.45);
\draw (1.75,0) rectangle (2.35,1.45);
\node[font=\scriptsize] at (2.05,-0.17) {blockage};
\draw[thick] (0.45,0) -- (0.45,0.95); \draw[thick] (0.2,0) -- (0.7,0);
\draw[fill=white] (0.3,0.95) rectangle (0.6,1.55);
\foreach \y in {1.05,1.17,...,1.5} \fill (0.45,\y) circle (0.028);
\node at (0.45,-0.17) {BS};
\draw[line width=1.3pt] (2.75,2.45) -- (5.55,2.45);
\foreach \x in {2.85,3.1,3.6,3.85,4.4,4.9,5.15,5.45}
  \draw[fill=white] (\x-0.07,2.38) rectangle (\x+0.07,2.52);
\draw[thin] (2.85,2.52) to[bend left=60] (3.6,2.52);
\draw[thin] (3.1,2.52) to[bend left=60] (4.4,2.52);
\draw[thin] (3.85,2.52) to[bend left=60] (4.9,2.52);
\draw[thin] (4.4,2.52) to[bend left=60] (5.45,2.52);
\draw[<->] (3.42,2.24) -- (3.78,2.24);
\draw[<->] (4.97,2.24) -- (5.33,2.24);
\draw[|<->|] (2.75,3.12) -- (5.55,3.12) node[midway,fill=white,inner sep=1pt] {$D$};
\node[right,align=left,font=\scriptsize] at (5.6,2.55) {Movable\\BD-RIS\\ Elements};
\draw[densely dotted] (4.15,2.35) -- (4.15,0.95);
\coordinate (T) at (3.35,0.55);
\coordinate (U) at (5.4,0.6);
\draw[->,thick] (0.62,1.42) -- (2.95,2.36) node[pos=0.5,above,sloped] {$\vect{G}$};
\draw[->,thick,red] (4.05,2.36) -- ($(T)+(0.08,0.2)$) node[pos=0.55,left] {$\htt$};
\draw[->,thick,blue] (4.25,2.36) -- ($(U)+(-0.1,0.22)$) node[pos=0.55,right] {$\hu$};
\draw[->,red,dashed] ($(T)+(-0.12,0.22)$) -- (3.86,2.36);
\draw[->,red,dashed] (2.95,2.25) -- (0.65,1.3) node[pos=0.55,below,sloped,font=\scriptsize] {echo};
\draw[->] (4.15,1.55) arc[start angle=-90,end angle=-58,radius=0.8];
\node[font=\scriptsize] at (4.62,1.48) {$\theta_{\rm u}$};
\draw[->] (4.15,1.75) arc[start angle=-90,end angle=-112,radius=0.6];
\node[font=\scriptsize] at (3.83,1.62) {$\theta_{\rm t}$};
\draw[red,thick] ($(T)+(-0.18,0)$) -- ($(T)+(0.18,0)$);
\draw[red] ($(T)+(-0.18,0.05)$) ellipse (0.09 and 0.025);
\draw[red] ($(T)+(0.18,0.05)$) ellipse (0.09 and 0.025);
\fill[red] ($(T)+(0,-0.03)$) rectangle ($(T)+(0.06,0.03)$);
\node[below,font=\scriptsize] at ($(T)+(0,-0.02)$) {target};
\draw[blue,thick,rounded corners=1pt] ($(U)+(-0.09,-0.18)$) rectangle ($(U)+(0.09,0.18)$);
\node[below,font=\scriptsize] at ($(U)+(0,-0.18)$) {user};
\end{tikzpicture}
\vspace{-1mm}
\caption{A movable BD-RIS serves a user and illuminates a target whose direct links to the BS are blocked; the target echo returns to the BS over the same surface.}
\label{fig:model}
\vspace{-3mm}
\end{figure}

In Fig.~\ref{fig:model}, a BS with $N$ antennas serves a single-antenna user and senses a point target through an $M$-element BD-RIS, since the direct links are blocked. The elements slide along a linear track of length $D\ge(M-1)d_0$, with positions $\vect{x}=[x_1,\dots,x_M]^\T$ in
\begin{equation}
\mathcal X_D=\{\vect{x}:\,x_1\ge 0,\; x_{m+1}-x_m\ge d_0,\; x_M\le D\},
\label{eq:X}
\end{equation}
where $d_0\le\lambda/2$ is the minimum spacing. Let $[\vect{e}(\vect{x},\theta)]_m=e^{j2\pi x_m\sin\theta/\lambda}$. The BS--RIS link is LoS, $\vect{G}=\bar\alpha\,\vect{e}(\vect{x},\varphi)\vect{a}^\Hm$, where $\vect{a}\in\mathbb C^N$ is the BS steering vector. The RIS--target link is $\htt=\beta_{\rm t}\vect{e}(\vect{x},\theta_{\rm t})$, and the RIS--user link is $\hu=\beta_{\rm u}\vect{e}(\vect{x},\theta_{\rm u})$ unless stated otherwise. We assume $\theta_{\rm u}\neq\theta_{\rm t}$; otherwise the two channels are collinear for every placement and no trade-off exists.

The scattering matrix $\vect\Theta$ is lossless and reciprocal, $\vect\Theta=\vect\Theta^\T$ and $\vect\Theta^\Hm\vect\Theta=\vect I$, and is realized by a real symmetric reactance matrix $\vect X$ as $\vect\Theta=(j\vect X+Z_0\vect I)^{-1}(j\vect X-Z_0\vect I)$~\cite{shen2022modeling}. A group-connected BD-RIS splits the elements into groups $\mathcal G$ of size $M_{\rm g}$, making $\vect\Theta$ block diagonal; $M_{\rm g}=M$ is the FC BD-RIS and $M_{\rm g}=1$ the diagonal RIS. Unlike~\cite{xu2026movable}, where interconnected elements move as one unit, every element moves individually and the lossless network is assumed to follow it\footnote{We remark that this is an idealization: position-dependent coupling, interconnection losses and actuation are not modeled and left for future work.}. 

The BS transmits $\vect s(t)$ with covariance $\vect R$, under the power budget $\operatorname{tr}(\vect R)\le P$. The user receives $y_{\rm u}=\hu^\T\vect\Theta\vect G\vect s+n_{\rm u}$, $n_{\rm u}\sim\mathcal{CN}(0,\sigma_{\rm u}^2)$. By reciprocity, the target echo at the BS is
\begin{equation}
\vect y_{\rm s}=\xi\,\vect G^\T\vect\Theta\htt\htt^\T\vect\Theta\vect G\vect s+\vect n_{\rm s},
\label{eq:echo}
\end{equation}
with reflection coefficient $\xi$ and $\vect n_{\rm s}\sim\mathcal{CN}(\vect 0,\sigma_{\rm s}^2\vect I)$.

\begin{lemma}[Transmit design]\label{lem:tx}
For $\vect G=\bar\alpha\vect e\vect a^\Hm$ and any $\vect\Theta$ and $\vect x$, a single stream $\vect s=\vect w s$ with maximum-ratio transmission (MRT), $\vect w=\sqrt P\vect a/\|\vect a\|$, maximizes both the user SNR and the echo SNR.
\end{lemma}
\begin{IEEEproof}
Write $\vect R=\vect R_{\rm d}+\vect R_{\rm s}$, where $\vect R_{\rm d}$ carries the user data and $\vect R_{\rm s}$ dedicated sensing streams. The user receives the useful power $|\bar\alpha|^2|\hu^\T\vect\Theta\vect e|^2\vect a^\Hm\vect R_{\rm d}\vect a$ and the interference $|\bar\alpha|^2|\hu^\T\vect\Theta\vect e|^2\vect a^\Hm\vect R_{\rm s}\vect a$, while the echo energy in~\eqref{eq:echo} is proportional to $\vect a^\Hm\vect R\vect a$. Hence $\vect R_{\rm s}=\vect 0$, and $\vect a^\Hm\vect R_{\rm d}\vect a\le P\|\vect a\|^2$ with equality iff $\vect R_{\rm d}=P\vect a\vect a^\Hm/\|\vect a\|^2$.
\end{IEEEproof}

Lemma~\ref{lem:tx} is a direct consequence of the rank-one BS--RIS link; it shows that all spatial processing happens at the surface. The incident field is $\vect g=\bar\alpha\sqrt{P}\|\vect a\|\,\vect e(\vect x,\varphi)$, with $\|\vect g\|^2=P_{\rm r}\triangleq|\bar\alpha|^2PNM$. With the reflected field $\vect u=\vect\Theta\vect g$, the user SNR and the target illumination are
\begin{equation}
\gc=|\hu^\T\vect u|^2/\sigma_{\rm u}^2,\qquad \Gt=|\htt^\T\vect u|^2.
\label{eq:metrics}
\end{equation}
Since $\vect e^\T\vect\Theta\htt=\htt^\T\vect\Theta\vect e$ for $\vect\Theta=\vect\Theta^\T$, the echo becomes $\vect y_{\rm s}=\xi\bar\alpha^2(\htt^\T\vect\Theta\vect e)^2(\vect a^\Hm\vect w)\,\vect a^*s+\vect n_{\rm s}$. Combining with $\mathbf a/\|\mathbf a\|$ over
$T$ known snapshots gives
$\mathrm{SNR}_{\rm s}
=|\xi|^2 T G_{\rm t}^2/(P\sigma_{\rm s}^2)$.  For each placement $\vect x$, we study the achievable region $\mathcal R(\vect x)=\{(\gc,\Gt):\vect u=\vect\Theta\vect g,\ \vect\Theta\in\mathcal U\}$, where $\mathcal U$ is the feasible set of $\vect\Theta$ for the given architecture, its Pareto boundary, and the placement $\vect x\in\mathcal X_D$ that is best for it.

\section{Analysis}
\subsection{Reachable Fields and the Pareto Boundary}
\begin{lemma}\label{lem:reach}
For a group-connected BD-RIS, $\{\vect\Theta\vect g:\vect\Theta\in\mathcal U\}=\{\vect u:\|\vect u_{\mathcal G}\|=\|\vect g_{\mathcal G}\|\ \forall\mathcal G\}$. In particular, an FC BD-RIS reaches every $\vect u$ with $\|\vect u\|=\|\vect g\|$, and a diagonal RIS every $\vect u$ with $|u_m|=|g_m|$.
\end{lemma}
\begin{IEEEproof}
Necessity follows from the unitarity of each block. For sufficiency, take one block and drop the index. Choose $\vect z$ with $\|\vect z\|=\|\vect g\|$ and $\vect z^\T\vect z=\vect g^\T\vect u$, e.g., $\vect z=\|\vect g\|e^{j\phi}[\cos\tau,\ j\sin\tau,\ 0,\dots,0]^\T$ with $\phi=\angle(\vect g^\T\vect u)/2$ and $\cos2\tau=|\vect g^\T\vect u|/\|\vect g\|^2\le1$. The pairs $(\vect z^*,\vect z)$ and $(\vect g^*,\vect u)$ have the same Gram matrix, so a unitary $\vect Q$ exists with $\vect Q\vect z^*=\vect g^*$ and $\vect Q\vect z=\vect u$. Then $\vect\Theta=\vect Q\vect Q^\T$ is symmetric and unitary, and $\vect\Theta\vect g=\vect Q(\vect Q^\Hm\vect g^*)^*=\vect Q\vect z=\vect u$.
\end{IEEEproof}

Lemma~\ref{lem:reach} concerns ideal scattering matrices. The reactance $\vect X=-jZ_0(\vect I+\vect\Theta)(\vect I-\vect\Theta)^{-1}$ is finite iff $1$ is not an eigenvalue of $\vect\Theta$. When it is, e.g., for $\vect u=\vect g$, the field is approached with growing reactances. Since~\eqref{eq:metrics} depends on $\vect u$ only up to a common phase, $\vect u$ can be rotated, and $\vect Q$ is free on the orthogonal complement.

Let $\vect p=\hu^*/\|\hu\|$ and $\vect q=\htt^*/\|\htt\|$, so that $\gc=\gcb|\vect p^\Hm\vect u|^2/P_{\rm r}$ and $\Gt=\Gtb|\vect q^\Hm\vect u|^2/P_{\rm r}$, where $\gcb=P_{\rm r}\|\hu\|^2/\sigma_{\rm u}^2$ and $\Gtb=P_{\rm r}\|\htt\|^2$ are the single-function maxima. Define
\begin{equation}
\rho=|\vect p^\Hm\vect q|=\frac{|\htt^\Hm\hu|}{\|\hu\|\|\htt\|},\qquad \psi=\arccos\rho\in[0,\pi/2].
\label{eq:rho}
\end{equation}

\begin{theorem}\label{thm:pareto}
For an FC BD-RIS and arbitrary $\hu,\htt$ with $0<\rho<1$, the Pareto boundary of $\mathcal R(\vect x)$ is
\begin{equation}
\big\{\big(\gcb\cos^2\vartheta,\ \Gtb\cos^2(\psi-\vartheta)\big):\ 0\le\vartheta\le\psi\big\},
\label{eq:pareto}
\end{equation}
attained by $\vect u^\star(\vartheta)=\sqrt{P_{\rm r}}\,[\sin(\psi-\vartheta)\vect p+\sin\vartheta\,\tilde{\vect q}]/\sin\psi$, with $\tilde{\vect q}=e^{-j\angle(\vect p^\Hm\vect q)}\vect q$, and $\vect\Theta$ from Lemma~\ref{lem:reach}. Formula~\eqref{eq:pareto} also holds for $\rho=0$, with $\tilde{\vect q}=\vect q$, and for $\rho=1$, where the boundary is the single point $(\gcb,\Gtb)$ attained by $\vect u=\sqrt{P_{\rm r}}\vect p$.
\end{theorem}
\begin{IEEEproof}
Fix $|\vect p^\Hm\vect u|^2=P_{\rm r}\cos^2\vartheta$ and write $\vect u=\sqrt{P_{\rm r}}(\cos\vartheta\,e^{j\nu}\vect p+\sin\vartheta\,\vect r)$ with $\vect r\perp\vect p$, $\|\vect r\|=1$. Write $\tilde{\vect q}=\rho\vect p+\sqrt{1-\rho^2}\,\vect p_\perp$ with $\vect p_\perp\perp\vect p$ of unit norm. Then $|\tilde{\vect q}^\Hm\vect u|/\sqrt{P_{\rm r}}\le\rho\cos\vartheta+\sqrt{1-\rho^2}\sin\vartheta=\cos(\psi-\vartheta)$, with equality iff $\vect r=e^{j\nu}\vect p_\perp$, i.e., $\vect u=e^{j\nu}\vect u^\star(\vartheta)$. Points with $\vartheta>\psi$ are dominated by $\vartheta=\psi$.
\end{IEEEproof}

This is the two-direction structure of the MISO Pareto boundary~\cite{jorswieck2008complete}. The new element is Lemma~\ref{lem:reach}, which makes a passive FC BD-RIS behave like a fully digital transmitter.

\begin{corollary}\label{cor:closed}
(a) For $\eta\in(0,1]$ and a requirement $\gc\ge\eta\gcb$, the maximum illumination is $\Gt^\star=\Gtb\cos^2\!\big([\psi-\arccos\sqrt\eta\,]^+\big)$, which equals $\Gtb$ iff $\rho\ge\sqrt\eta$. (b) $\mathcal R(\vect x)$ is convex, being a scaled joint numerical range of $\vect p\vect p^\Hm$ and $\vect q\vect q^\Hm$~\cite{horn1991topics}, and $\max\, w\gc/\gcb+(1-w)\Gt/\Gtb=\tfrac12\big[1+\sqrt{1-4w(1-w)(1-\rho^2)}\big]$. Both increase with $\rho$.
\end{corollary}
\begin{IEEEproof}
(a) $\gc\ge\eta\gcb$ means $\vartheta\le\arccos\sqrt\eta$ in~\eqref{eq:pareto}, and $\Gt$ increases with $\vartheta$ on $[0,\psi]$. (b) The maximum is the largest eigenvalue of $w\vect p\vect p^\Hm+(1-w)\vect q\vect q^\Hm$, whose nonzero eigenvalues are those of a $2\times2$ Gram matrix with trace $1$ and determinant $w(1-w)(1-\rho^2)$.
\end{IEEEproof}

For group-connected and diagonal RISs, Lemma~\ref{lem:reach} turns the weighted sum into $\max\vect u^\Hm\vect B_w\vect B_w^\Hm\vect u$ over a product of spheres, with $\vect B_w=[\sqrt w\,\vect p,\ \sqrt{1-w}\,\vect q]$.
\begin{proposition}\label{prop:group}
For any group partition,
\begin{equation}
\begin{aligned}
&\max_{\|\vect u_{\mathcal G}\|=\|\vect g_{\mathcal G}\|\,\forall\mathcal G}\vect u^\Hm\vect B_w\vect B_w^\Hm\vect u\\
&\qquad=\max_{\vect z\in\mathbb C^2,\,\|\vect z\|=1}\Big(\sum_{\mathcal G}\|\vect g_{\mathcal G}\|\,\|[\vect B_w\vect z]_{\mathcal G}\|\Big)^2,
\end{aligned}
\label{eq:prop1}
\end{equation}
and a maximizer is $\vect u_{\mathcal G}=\|\vect g_{\mathcal G}\|\vect v_{\mathcal G}/\|\vect v_{\mathcal G}\|$ with $\vect v=\vect B_w\vect z^\star$, where any $\vect u_{\mathcal G}$ with $\|\vect u_{\mathcal G}\|=\|\vect g_{\mathcal G}\|$ may be used if $\vect v_{\mathcal G}=\vect 0$. Since $\vect z$ has two real degrees of freedom up to a common phase, every supported point of the boundary follows from a two-dimensional search.
\end{proposition}
\begin{IEEEproof}
$\vect u^\Hm\vect B_w\vect B_w^\Hm\vect u=\|\vect B_w^\Hm\vect u\|^2=\max_{\|\vect z\|=1}|(\vect B_w\vect z)^\Hm\vect u|^2$. For fixed $\vect z$, the Cauchy--Schwarz inequality on each block gives $|\sum_{\mathcal G}\vect v_{\mathcal G}^\Hm\vect u_{\mathcal G}|\le\sum_{\mathcal G}\|\vect v_{\mathcal G}\|\|\vect g_{\mathcal G}\|$, with equality at the stated $\vect u$. Exchanging the two maximizations gives~\eqref{eq:prop1}.
\end{IEEEproof}

When the region is nonconvex, as for the diagonal RIS in Fig.~\ref{fig:pareto}, its unsupported points are computed numerically: we maximize $\Gt$ under $\gc\ge\eta\gcb$ over the same family $\vect v=\vect B_w\vect z$ with local refinement, and check the result against random-restart SQP over all $M$ phases. We make no global-optimality claim for these points.

\subsection{Line of Sight: Positions Act Only Through $\rho$}
\begin{theorem}\label{thm:sep}
Under LoS, $\|\vect g\|$, $\|\hu\|$ and $\|\htt\|$ do not depend on $\vect x$, and
\begin{equation}
\rho(\vect x)=\frac1M\Big|\sum_{m=1}^M e^{j2\pi\Ds x_m/\lambda}\Big|,\qquad \Ds=\sin\theta_{\rm u}-\sin\theta_{\rm t},
\label{eq:rhoLoS}
\end{equation}
independently of the BS--RIS geometry. For an FC BD-RIS, if $\rho(\vect x)\le\rho(\vect x')$, every point of $\mathcal R(\vect x)$ is weakly Pareto-dominated by a point of $\mathcal R(\vect x')$. Hence every criterion that is nondecreasing in $\gc$ and $\Gt$ is maximized by
\begin{equation}
\vect x^\star=\arg\max_{\vect x\in\mathcal X_D}\ \rho(\vect x).
\label{eq:place}
\end{equation}
\end{theorem}
\begin{IEEEproof}
All entries have constant modulus, so the norms are fixed and~\eqref{eq:rhoLoS} follows from~\eqref{eq:rho}. Every point of $\mathcal R(\vect x)$ is dominated by a point of~\eqref{eq:pareto}, and for $\psi'\le\psi$ the point $\vartheta$ of~\eqref{eq:pareto} is dominated by the point $\min(\vartheta,\psi')$ of the boundary for $\psi'$.
\end{IEEEproof}

Dominance cannot be replaced by set inclusion: for $\rho'=1$, $\mathcal R(\vect x')$ is the segment $\{(t\gcb,t\Gtb):0\le t\le1\}$, which does not contain $(\gcb,\rho^2\Gtb)$.

\begin{theorem}\label{thm:lattice}
Let $\Lambda=\lambda/|\Ds|$. Then $\rho(\vect x)=1$ iff $x_m\in x_1+\Lambda\mathbb Z$ for all $m$, and
\begin{equation}
\max_{\vect x\in\mathcal X_D}\rho(\vect x)=1\iff D\ge D^\star=\frac{(M-1)\lambda}{|\Ds|},
\label{eq:Dstar}
\end{equation}
attained by $x_m=(m-1)\Lambda$. At this placement, the Pareto boundary of the FC BD-RIS reduces to the single point $(\gcb,\Gtb)$: the user and the target are served as if each were alone.
\end{theorem}
\begin{IEEEproof}
Equality in the triangle inequality for~\eqref{eq:rhoLoS} holds iff all phases $2\pi\Ds x_m/\lambda$ coincide modulo $2\pi$. Since $|\Ds|\le2$, $\Lambda\ge\lambda/2\ge d_0$, so $M$ ordered lattice points are feasible iff their span $(M-1)\Lambda$ fits in $D$. With $\psi=0$, Theorem~\ref{thm:pareto} gives $(\gcb,\Gtb)$.
\end{IEEEproof}

\begin{proposition}[Clustered lattice]\label{prop:cluster}
Let $c$ divide $M$ with $cd_0\le\Lambda$. Placing $M/c$ clusters of $c$ elements, spaced $d_0$ within each cluster, at the lattice points $n\Lambda$ achieves
\begin{equation}
\rho_c=\left|\frac{\sin(c\pi\Ds d_0/\lambda)}{c\sin(\pi\Ds d_0/\lambda)}\right|\ \ \text{with}\ \ D_c=\Big(\frac Mc-1\Big)\Lambda+(c-1)d_0 .
\label{eq:cluster}
\end{equation}
Hence $\max_{\vect x\in\mathcal X_D}\rho(\vect x)\ge\rho_c$ for all $D\ge D_c$, and $c=1$ recovers Theorem~\ref{thm:lattice}.
\end{proposition}
\begin{IEEEproof}
The sum in~\eqref{eq:rhoLoS} factorizes into the cluster array factor, which gives~\eqref{eq:cluster}, and the lattice factor, which has unit modulus. Clusters do not overlap because $cd_0\le\Lambda$.
\end{IEEEproof}

\begin{corollary}[Diagonal RIS]\label{cor:moving}
Under LoS, $|g_m|^2=P_{\rm r}/M$, so a diagonal RIS realizes every constant-modulus $\vect u$. (a) The endpoints of~\eqref{eq:pareto}, $\vect u\propto\vect p$ and $\vect u\propto\tilde{\vect q}$, have constant modulus, so a diagonal RIS attains them for any $\vect x$. (b) For $D\ge D^\star$, a movable diagonal RIS on the lattice attains $(\gcb,\Gtb)$, and interconnections bring no further gain.
\end{corollary}

For $D<D^\star$, the diagonal region depends on all the phases $2\pi\Ds x_m/\lambda$, not only on $\rho$: for $M=4$, $d_0=\lambda/2$ and $\Ds=1/2$, the placements $(0,2,4,5)\lambda$ and $(0,2,2.5,3.5)\lambda$ both give $\rho=1/2$, but their best values of $\gc/\gcb+\Gt/\Gtb$ are $1.25$ and $1.457$. The placement~\eqref{eq:place} is therefore optimal for the FC BD-RIS only. For diagonal and group-connected RISs below $D^\star$, we use it as an achievable design.

\begin{corollary}[Minimum aperture]\label{cor:aperture}
Let $\rho_{\max}(D)=\max_{\vect x\in\mathcal X_D}\rho(\vect x)$. The FC BD-RIS meets $\gc\ge\eta\gcb$ without any sensing loss iff $D\ge D_\eta=\min\{D:\rho_{\max}(D)\ge\sqrt\eta\}$. Moreover, $D_\eta\le D^\star$, with equality for $\eta=1$, and $\rho_{\max}(D)$ is nondecreasing because $\mathcal X_D\subseteq\mathcal X_{D'}$ for $D\le D'$.
\end{corollary}


\begin{remark}[Position errors]
For i.i.d.\ errors $e_m\sim\mathcal N(0,\sigma^2)$ on $x_m$, Jensen's inequality gives $\mathbb E[\rho]\ge\rho(\vect x)\exp(-2\pi^2\Ds^2\sigma^2/\lambda^2)$. For $\sigma=\lambda/20$ and $\Ds=0.52$, the factor is $0.987$, so the lattice is robust to small positioning errors.
\end{remark}

\subsection{Placement for $D<D^\star$ and Beyond LoS}
For a smaller aperture,~\eqref{eq:place} is nonconvex. Using $|z|=\max_\omega\Re\{e^{-j\omega}z\}$,
\begin{equation}
M\max_{\vect x\in\mathcal X_D}\rho(\vect x)=\max_{\omega}\max_{\vect x\in\mathcal X_D}\sum_{m=1}^M\cos(2\pi\Ds x_m/\lambda-\omega).
\label{eq:omega}
\end{equation}
For fixed $\omega$, the inner problem is separable apart from the ordering constraints in~\eqref{eq:X}. On a grid of step $\delta$ with $d_0=K\delta$ and $f_i=\cos(2\pi\Ds i\delta/\lambda-\omega)$, the recursion
\begin{equation}
V_1(i)=f_i,\qquad V_m(i)=f_i+\max_{j\le i-K}V_{m-1}(j),
\label{eq:dp}
\end{equation}
evaluated with running maxima, solves the inner grid problem exactly in $\mathcal O(MD/\delta)$ operations.

\begin{proposition}[Accuracy]\label{prop:dp}
Let $\hat\rho$ be the best $\rho$ returned by~\eqref{eq:dp} over $N_\omega$ equispaced values of $\omega$. Then $\hat\rho\ge\cos(\pi/N_\omega)\,[\rho_{\max}(D)-2\pi|\Ds|\delta/\lambda]$.
\end{proposition}
\begin{IEEEproof}
Rounding every position of a continuous optimizer down to the grid keeps it in $\mathcal X_D$, because $d_0=K\delta$, and changes each phase by at most $2\pi|\Ds|\delta/\lambda$. The grid optimum therefore loses at most $2\pi|\Ds|\delta/\lambda$ in $\rho$. For the grid optimum $\vect x^\circ$, some $\omega_k$ lies within $\pi/N_\omega$ of $\angle\sum_m e^{j2\pi\Ds x^\circ_m/\lambda}$, and the DP at $\omega_k$ returns $\vect x$ with $M\rho(\vect x)\ge M\rho(\vect x^\circ)\cos(\pi/N_\omega)$.
\end{IEEEproof}

We use $N_\omega=96$ with local refinement around the best $\omega$, so the factor is $0.9995$. For $D\ge D^\star$, the analytical lattice of Theorem~\ref{thm:lattice} is used directly.

For non-LoS channels, Theorem~\ref{thm:pareto}, Corollary~\ref{cor:closed} and Proposition~\ref{prop:group} still hold, but $\gcb$ depends on $\vect x$. For an absolute requirement $\gc\ge\Gamma$, Corollary~\ref{cor:closed}(a) with $\eta(\vect x)=\min(1,\Gamma/\gcb(\vect x))$ gives a closed-form objective in $\rho(\vect x)$ and $\|\hu(\vect x)\|$. If $\Gamma>\gcb(\vect x)$, the user is served at $\gcb(\vect x)$. We improve this objective by exact grid updates of each $x_m$ between its neighbours, initialized at the lattice, accepting only strict improvements. The search is monotone and terminates, but it is local. Algorithm~\ref{alg:design} summarizes the design.

\begin{algorithm}[!t]
\caption{Movable BD-RIS ISAC design}
\label{alg:design}
\begin{algorithmic}[1]
\Require Angles $\theta_{\rm u}\neq\theta_{\rm t}$,
aperture $D$, requirement $\eta$

\State \textbf{Placement:}
If $D\ge D^\star$, set
$x_m=(m-1)\lambda/|\Ds|$
(Theorem~\ref{thm:lattice});
otherwise solve~\eqref{eq:omega}
by the DP~\eqref{eq:dp}.
For non-LoS channels, refine by coordinate ascent.

\State \textbf{Transmitter:}
MRT toward the RIS (Lemma~\ref{lem:tx}).

\State \textbf{Field:}
FC:
$\vect u=\vect u^\star(
\min(\psi,\arccos\sqrt\eta))$
(Theorem~\ref{thm:pareto});
group-connected or diagonal:
search over $\vect z$
in~\eqref{eq:prop1}.

\State \textbf{Scattering matrix:}
Synthesize $\vect\Theta$ and $\vect X$
blockwise by Lemma~\ref{lem:reach}.

\end{algorithmic}
\end{algorithm}

\emph{Complexity:} Step~1 costs $\mathcal O(N_\omega MD/\delta)$, step~3 costs $\mathcal O(M)$ for the FC BD-RIS and $\mathcal O(N_{\rm s}M)$ for $N_{\rm s}$ search points otherwise, and step~4 costs $\mathcal O(MM_{\rm g}^2)$ over all groups. Unlike~\cite{esmaeilbeig2025bd,liu2024enhancing,ye2026joint,peng2026exploiting}, no alternating optimization is needed.

\section{Simulation Results}\label{sec:sim}
In this section, we asses the perforamnce of the proposed design. Table~\ref{tab:params} lists the simulation parameters. Under LoS, $\gcb$ and $\Gtb$ do not depend on $\vect x$, so we report metrics normalized by these single-function maxima; they hold for any transmit power $P$, number of BS antennas $N$, path loss and reflection coefficient $\xi$. Since $\mathrm{SNR}_{\rm s}\propto\Gt^2$, the sensing SNR loss is $-20\log_{10}(\Gt/\Gtb)$\,dB.

 \begin{table}[t]
\centering
\caption{Simulation Parameters}
\label{tab:params}
\scriptsize
\setlength{\tabcolsep}{3pt}
\renewcommand{\arraystretch}{1.05}
\begin{tabular}{@{}lll@{}}
\hline
\textbf{Parameter} & \textbf{Symbol} & \textbf{Value} \\
\hline
RIS elements & $M$ & $16$ ($32$ in Figs.~\ref{fig:rho}, \ref{fig:mismatch}, \ref{fig:rician}) \\
Minimum spacing & $d_0$ & $\lambda/2$ \\
Carrier frequency & $f_c$ & $28$\,GHz ($\lambda=1.07$\,cm) \\
BS--RIS angle of arrival & $\varphi$ & $-35^\circ$ \\
User / target angle & $\theta_{\rm u}$, $\theta_{\rm t}$ & $20^\circ$ / $-10^\circ$ ($[-60^\circ,15^\circ]$ in Fig.~\ref{fig:mismatch}) \\
Spatial-frequency offset & $\Ds$ & $0.516$ \\
Lattice spacing & $\Lambda$ & $1.94\lambda$ \\
Trade-off-free aperture & $D^\star$ & $29.1\lambda$ ($60.1\lambda$ for $M=32$) \\
Fixed baseline & -- & $\lambda/2$-spaced array \\
DP grid / phase search & $\delta$, $N_\omega$ & $\lambda/100$ / $96$ \\
User SNR loss tolerance & $\eta$ & $1$\,dB, $3$\,dB \\
Group size & $M_{\rm g}$ & $4$ \\
\hline
\multicolumn{3}{@{}l}{\emph{Rician RIS--user link (Fig.~\ref{fig:rician})}}\\
Rician factor & $K$ & $-5$ to $25$\,dB \\
Scattered paths & $L$ & $5$, angles uniform in $[-60^\circ,60^\circ]$ \\
User requirement & $\Gamma$ & $3$\,dB below LoS maximum \\
Position step / realizations & -- & $\lambda/50$ / $100$ \\
\hline
\end{tabular}
\vspace{-3mm}
\end{table}

\begin{figure}[!t]
    \centering
    \includegraphics[width=\columnwidth]{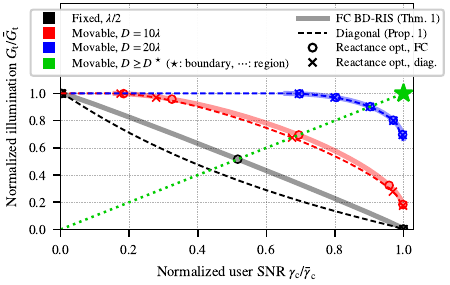}
    \vspace{-1.5mm}
    \caption{Pareto boundaries for two movable apertures
    and the fixed $\lambda/2$ array ($M=16$).
    Thick lines: FC (Theorem~\ref{thm:pareto});
    thin dashed: diagonal (Proposition~\ref{prop:group}),
    whose largest gap to FC is $0.16$, $0.04$ and $0.003$;
    markers: direct optimization of $\vect X$;
    green: $D\ge D^\star$.}
    \label{fig:pareto}
    \vspace{-3mm}
\end{figure}

\begin{figure*}[!t]
    \centering

    \subfloat[Sensing SNR loss versus aperture for user
    SNR losses of 1\,dB (solid) and 3\,dB (dashed),
    $M=16$.\label{fig:loss}]{%
        \includegraphics[
            width=0.32\textwidth,
            height=4.1cm,
            keepaspectratio
        ]{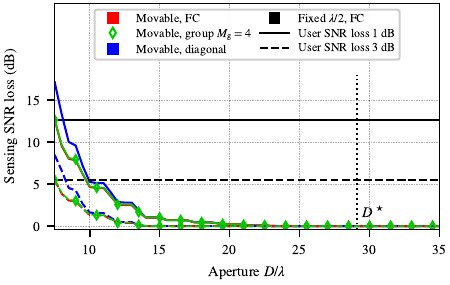}%
    }
    \hfill
    \subfloat[Normalized illumination at full user SNR
    versus target angle, at equal aperture
    $D=\lceil D^\star\rceil$ designed for
    $\theta_{\rm t}=-10^\circ$.
    \label{fig:mismatch}]{%
        \includegraphics[
            width=0.32\textwidth,
            height=4.1cm,
            keepaspectratio
        ]{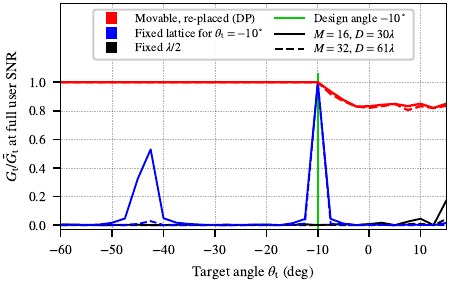}%
    }
    \hfill
    \subfloat[Average sensing SNR loss versus the Rician
    factor of the RIS--user link (absolute user requirement
    3\,dB below the LoS maximum, 100 realizations,
    95\% confidence intervals).
    \label{fig:rician}]{%
        \includegraphics[
            width=0.32\textwidth,
            height=4.1cm,
            keepaspectratio
        ]{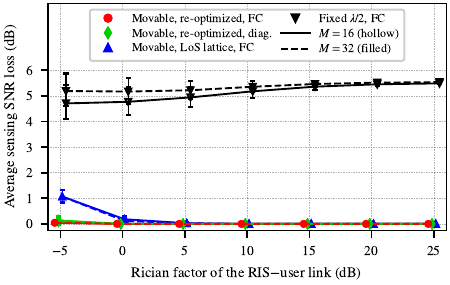}%
    }

    \vspace{-1.5mm}
    \caption{Sensing performance of the proposed architecture:
    (a) impact of aperture size and user SNR requirements;
    (b) target-angle mismatch;
    (c) robustness under Rician fading.}
    \label{fig:sensing_results}
    \vspace{-3mm}
\end{figure*}

\emph{Pareto boundaries:} Fig.~\ref{fig:pareto} compares $D=10\lambda$ and $D=20\lambda$ with the fixed $\lambda/2$ array. The FC curves are exact by Theorem~\ref{thm:pareto}. As a consistency check, the markers optimize $\vect X$ directly by L-BFGS from random starts, without using Lemma~\ref{lem:reach}. The FC markers agree with~\eqref{eq:pareto} to within $10^{-9}$. The diagonal markers (best of 150 starts) lie on the supported points from Proposition~\ref{prop:group}. The fixed array has nearly orthogonal channels ($\rho=0.033$); movement raises $\rho$ to $0.39$ and $0.81$, pushing the boundary outward. For $D\ge D^\star$, the Pareto boundary of both the FC and the diagonal RIS is the corner $(1,1)$ (green star), while the achievable region is the green dotted segment, as Theorem~\ref{thm:lattice} predicts. The diagonal boundary touches the FC boundary at its endpoints (Corollary~\ref{cor:moving}); for the fixed array it is nonconvex. The figure thus separates the two degrees of freedom: interconnection matters only at fixed positions, where the largest gap between the FC and diagonal boundaries is $0.16$, whereas movement shrinks this gap to $0.04$ and $0.003$ at $D=10\lambda$ and $20\lambda$ and removes it at $D\ge D^\star$, while enlarging the region of both architectures far more.

\begin{figure}[t]
\centering
\includegraphics{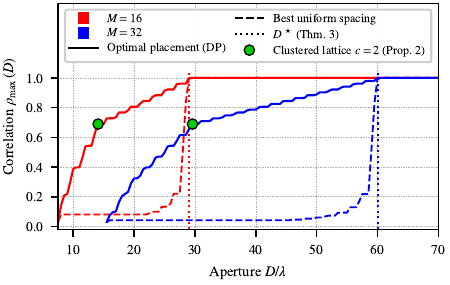}
\vspace{-1.5mm}
\caption{Maximum correlation versus aperture for $M\in\{16,32\}$; dotted lines mark $D^\star$ (Theorem~\ref{thm:lattice}), green circles Proposition~\ref{prop:cluster}.}
\label{fig:rho}
\vspace{-3mm}
\end{figure}

\emph{Moving vs.\ interconnecting:} Fig.~\ref{fig:loss} shows the sensing SNR loss when the user may lose 1\,dB or 3\,dB of SNR. The fixed FC BD-RIS loses $12.7$ and $5.5$\,dB; the fixed diagonal RIS (not shown) loses $17.1$ and $8.5$\,dB. With movement, the FC loss vanishes once $\rho\ge\sqrt\eta$ (Corollary~\ref{cor:aperture}), at $D=25\lambda$ and $15.5\lambda$. With the same placement, the diagonal RIS matches the FC BD-RIS from these apertures on and trails by at most $0.61$ and $0.29$\,dB for $D\ge10\lambda$. A group-connected RIS with $M_{\rm g}=4$ stays within $0.14$\,dB of the FC BD-RIS.

\emph{Geometry mismatch:} In Fig.~\ref{fig:mismatch}, the aperture is fixed at $D=\lceil D^\star\rceil$ for the nominal angle, and the user must keep its full SNR ($\eta=1$, so $\Gt/\Gtb=\rho^2$). The comparison isolates the value of movability at equal aperture: a fixed sparse array on the nominal lattice is optimal at $\theta_{\rm t}=-10^\circ$, but away from it its illumination stays below $0.53$ ($M=16$) and $0.03$ ($M=32$), close to the $\lambda/2$ array. The movable surface, re-placed by the DP for each $\theta_{\rm t}$, keeps $\Gt=\Gtb$ for $\theta_{\rm t}\le-10^\circ$, where $D\ge D^\star$, and above $0.80\,\Gtb$ when the target moves toward the user and $D<D^\star$.

\emph{Rician fading:} In Fig.~\ref{fig:rician}, the RIS--user link has $L=5$ scattered paths at angles uniform in $[-60^\circ,60^\circ]$, and the user requires an absolute SNR 3\,dB below its LoS maximum. Results are averaged over 100 realizations with 95\% confidence intervals. The LoS lattice of Theorem~\ref{thm:lattice} loses $1.1$\,dB at a Rician factor of $-5$\,dB and less than $0.2$\,dB from $0$\,dB on. Coordinate-wise re-optimization reduces the loss to at most $0.05$\,dB (FC) and $0.14$\,dB (diagonal) for both $M$, while the fixed $\lambda/2$ array loses $4.7$--$5.5$\,dB. The requirement cannot be met in at most $3\%$ of the realizations with the re-optimized placement and $7\%$ with the fixed array; in those cases the user is served at its maximum.

\emph{Phase transition:} Fig.~\ref{fig:rho} plots $\rho_{\max}(D)$ for $M\in\{16,32\}$. The DP curves approach $1$ as $D\to D^\star$, where the analytical lattice gives exactly $\rho=1$. At $D=D^\star/2$ the DP attains $\rho=0.69$ for both $M$, whereas the best uniform spacing stays at $0.08$ ($M=16$) and $0.04$ ($M=32$). The plateau is the clustered lattice of Proposition~\ref{prop:cluster} with $c=2$ (circles): the DP returns pairs of elements on the lattice and attains $\rho_2=0.6895$ at $D_2=14.1\lambda$ ($29.6\lambda$ for $M=32$).

\section{Conclusion}
For a movable BD-RIS serving one user and illuminating one target, the communication--sensing Pareto boundary is available in closed form and is governed by a single quantity, the correlation between the user and target channels at the surface. Under LoS, the placement that maximizes this correlation is optimal for the FC BD-RIS. A lattice of spacing $\lambda/|\Ds|$ steers a grating lobe at the target and reduces the Pareto boundary to the single-function optimum once $D\ge(M-1)\lambda/|\Ds|$; there, a movable diagonal RIS is also optimal. Below this aperture, a DP with certified accuracy and a clustered lattice give the placement, and the diagonal RIS remains within a fraction of a decibel of the FC BD-RIS at moderate apertures. The gains hold at equal aperture under geometry mismatch and, after local re-optimization, under Rician fading.  

\bibliographystyle{IEEEtran}
\bibliography{refs}
\end{document}